\documentclass[11pt,envcountsect,envcountsame]{llncs}
\pdfoutput=1
\usepackage{a4wide}
\usepackage{paralist} %

\usepackage[utf8]{inputenc}
\usepackage{amsmath}
\usepackage{amssymb}
\usepackage{xspace} %
\usepackage{tikz} %

\newcommand{\macro}[2]{ \providecommand{#1}{{\ensuremath{#2}}\xspace}}

\macro{\N}{\mathbb N}
\macro{\noir}{\bullet}
\macro{\blanc}{\circ}
\macro{\lnoir}{\blanc\!\!\rightarrow\!\!\noir}
\macro{\lblanc}{\blanc\!\!\leftarrow\!\!\noir}
\macro{\lall}{\blanc\;\;\;\noir}
\macro{\lok}{\blanc\!\!\leftrightarrows\!\!\noir}
\macro{\proc}{x}%
\macro{\scheme}{\mathcal S}
\macro{\algo}{\mathcal A}
\macro{\prob}{\mathcal P}

\begin{document}
\title{Consensus vs Broadcast in Communication Networks 
with Arbitrary Mobile Omission Faults}
\author{Emmanuel Godard$^{1,2}$ and Joseph Peters$^2$%
\thanks{Research supported by NSERC of Canada}%
}
\institute{$^1$ Pacific Institute for Mathematical Sciences, CNRS UMI 3069\\
$^2$ School of Computing Science, Simon Fraser University 
}
\date{} 
\pagestyle{plain}
\maketitle

\begin{abstract}
We compare the solvability of the Consensus and Broadcast problems in
synchronous communication networks in which the delivery of messages is
not reliable. The failure model is the mobile omission faults model.
During each round, some messages can be lost and the set of
possible simultaneous losses is the same for each round. 
We investigate these problems for the first time for arbitrary sets of
possible failures. Previously, these sets were
defined by bounding the numbers of failures.
In this setting, we present a new necessary condition for the solvability
of Consensus that unifies previous impossibility results in this
area. This condition is expressed using Broadcastability properties.
As a very important application, we show that when the
sets of omissions that can occur are defined by bounding the
numbers of failures, counted in \emph{any} way (locally, globally, etc.), 
then the Consensus problem is actually equivalent to the Broadcast
problem.
\end{abstract}

\setcounter{page}{1}
\setcounter{footnote}{0}

\section{Introduction}

We consider synchronous communication networks in which some messages can
be lost during each round. These omission faults can be permanent or not;
a faulty link can become reliable again after an unpredictable number of
rounds, and it can continue to alternate between being reliable and faulty
in an unpredictable way.
This model is more general than other models, such as \emph{component
failure} models, in which failures, once they appear somewhere, are
located there permanently.
The model that we use, called the \emph{mobile faults} or \emph{dynamic faults}
model, was introduced in \cite{timeisnotahealer} and is discussed further
in \cite{SW07}. 
An important property of these systems is that the set of possible
simultaneous omissions is the same for each round. In some sense, the
system has no ``memory'' of the previous failures. 
Real systems often exhibit such memory-less behaviour.
 
In previous research, the sets of possible simultaneous omissions were
defined by bounding the numbers of omissions. Recent work on this subject 
includes \cite{SW07}, in which omissions are counted globally, and 
\cite{SWK09}, in which the number of omissions is locally
bounded. It has also been shown to be good for \emph{layered analysis}
\cite{layering}. 
In this paper we consider the \emph{most general case} of such systems, i.e.
systems in which the set of possible simultaneous omissions is arbitrary. 
This allows the modelling of any system in which omissions can happen
transiently, in any arbitrary pattern, including systems in which the
communications are not symmetric.    
 
We investigate two fundamental problems of Distributed Computing in
these networks: the Consensus problem and the Broadcast
problem. While it has long been known that solvability of the
Broadcast problem implies solvability of the Consensus problem, we
prove here that these problems are actually equivalent (from both the
solvability and complexity points of view) when the sets of possible
omissions are defined by bounding the number of failures, for \emph{any}
possible way of counting them (locally, globally,
any combination, etc.).

\subsection{The Consensus Problem}
The Consensus problem is a very well studied problem in the area of
Distributed Algorithms. It is defined as follows. Each node of the
network starts with an initial value, and all nodes of the network
have to agree on a common value, which is one of the initial values.
Many versions of the problem concern the design of algorithms for
systems that are unreliable.

The Consensus problem has been widely studied in the context of
shared memory systems and message passing systems in which any node
can communicate with any other node.
Surprisingly, there have been few studies in the context of
communication networks, where the communication graph is not a
complete graph. In one of the first thorough studies 
\cite{SW07}, Santoro and Widmayer investigate some
$k-$Majority Problems that are defined as follows.
Each node starts with an initial value, and every node has to compute
a final \emph{decided} 
value such that there exists a value (from the set of initial values)
that is decided by at least $k$ of the nodes. The Consensus problem
(called the Unanimity problem in \cite{SW07}\footnote{%
Note that some of the terminology that we use in this paper is
different from the terminology of Santoro and Widmayer. We are
investigating relationships between different areas of distributed
algorithms, and some terminology (such as $k-$agreement) has
different meanings in the different areas.}) is
the $n-$Majority problem where $n$ is the number of nodes in the
network. 

In their paper, Santoro and Widmayer give results about
solving the Consensus problem in communication networks with various
types of faults including omission
faults. For simplicity, we focus here only on omission faults. 
We believe that our
results can be quite easily extended to other fault models,
using the methodology of \cite{SW07}.

\subsection{The Broadcast Problem}
Two of the most widely studied patterns of information propagation in
communication networks are \emph{broadcasting} and \emph{gossiping}. 
A \emph{broadcast} is the distribution of an initial value from one
node of a network to every other node of the network. A \emph{gossip}
is a simultaneous broadcast from every node of the network. The Broadcast
problem that we study in this paper is to find a node from which a
broadcast can be successfully completed.

There are close relationships between broadcasting and gossiping, and
the Consensus problem.
Indeed, the Consensus problem can be solved by first gossiping and
then applying a deterministic function at each node to the set of
initial values. But a gossip is not actually necessary.
If there exists a distinguished node $v_0$ in the network, then a
Consensus algorithm can be easily derived from an algorithm that
broadcasts from $v_0$. However the Broadcast problem and the Consensus
problem are not equivalent, as will be made clearer in
Section~\ref{theproblems}.

\subsection{Our Contributions}

In this paper, we investigate systems in which
the pattern of omission failures is arbitrary. 
A set of simultaneous omissions is called a \emph{communication event}.
We characterize the solvability of the Broadcast and Consensus problems
subject to an arbitrary family of possible communication
events. A node from
which it is possible to broadcast if the system is restricted to a given
communication event is called a \emph{source} for the communication event.
We prove that the Broadcast problem is solvable if
and only if there exists a common source for all communication events.
To study the Consensus problem, we define an 
equivalence relation on a family of communication events 
based on the collective local observations of the events by the
sources.  
We prove in Theorem~\ref{impossibility} that
the Consensus problem is not solvable for a family of
communication events 
if the Broadcast problem is not solvable for one class of the
equivalence relation. 
\emph{For Consensus to be solvable, the sources of a given event must be
  able to collectively distinguish communication events with
  incompatible sources.}  
We conjecture that this is actually a sufficient condition.

It is very
simple to characterize Broadcastability (see
Theorem~\ref{broadcastability}), so we get very simple and efficient
impossibility proofs for solving Consensus subject to arbitrary
omission failures.
These impossibility conditions are satisfied by the omission schemes of
\cite{SW07} and \cite{SWK09}. 
This means that our results encompass all previous known
results in the area. 

Furthermore, we prove that under very general conditions, in
particular when the possible simultaneous omissions are defined  
by bounding the number of omissions, for any way of counting omissions,
there is actually only one equivalence class when the
system is not broadcastable. 
An important application is that,
\emph{the Consensus problem is exactly the same as the Broadcast problem
for most omission fault models}. 
Therefore, it is possible to deduce complexity
results for the Consensus problem from complexity results about
broadcasting with omissions. 

\subsection{Related Work}

In \cite{HO}, %
the authors present a model that can
describe benign faults. This model is called the ``Heard-Of'' model.
It is a round-based model for an omission-prone environment in which the
set of possible communication events is not necessarily the same for
each round. 
However, they require a time-invariance property.
As a special case, they present a characterization that shows
that solving Consensus in this environment is equivalent to solving a
Selection problem. They also present algorithms for some %
families of omission schemes. It is not possible to derive our
simple characterizations from \cite{HO}.

In \cite{SW07}, the Consensus and related Agreement problems are
investigated for networks in which there are at most $f$ omissions
during any given round. 
It is proved that it is impossible to solve Consensus if
$f$ is at least the minimum degree of the graph. 
A Consensus algorithm is presented for the case where $f$ is strictly
smaller than the connectivity of the network.
In \cite{2generals}, a reduction to the two
process case is used to show that the connectivity of the graph is indeed the
exact limit for consensus is such omission schemes. In this paper,
we generalize
these results, showing that exact limits for Consensus can be
derived from exact limits for Broadcast. 
 
While the above studies use a global failure metric, a local failure
metric is investigated in \cite{SWK09}, distinguishing send and
receive omissions. 
The authors describe which bounds allow
Consensus to be solved, using a proof technique based on a \emph{Withholding
  Lemma}.
They claim that Consensus is solvable if and only if no node
can withhold its information from some other part of the network. 
We will show that this is not true when the pattern of omissions can be
arbitrary. 
We  present (in Example~\ref{contrex}) a system
in which nodes can withhold information infinitely but Consensus is solvable. 

Finally, it is worth noting that, although both \cite{SW07} and
\cite{SWK09} are using the classic bivalency proof technique, it
is not possible to derive any of the results of
\cite{SW07} or \cite{2generals} (global bound on omissions)
from the results of \cite{SWK09} (local bound) as
the omission schemes are not comparable. 
Our results consolidate these previous results.
Furthermore, our approach is more general than these previous approaches
and is more suitable for applications to new omission metrics.

\section{Definitions and Notation}
\subsection{Communication Networks}
We model a communication network by a digraph
$G=(V,E)$ which does not have to be symmetric.
If we are given an undirected graph $G$, we consider the corresponding
symmetric digraph.
We always assume that nodes have unique identities.
Given a set of arcs $E$, we define $h(E)=\{t\mid (s,t)\in E\}$,
the set of nodes that are heads of arcs in $E$.

All sub-digraphs that we consider in this paper are spanning subgraphs. Since
all spanning subgraphs of a digraph have the same set of nodes, we will
use the same notation to refer to both the set of arcs of a sub-digraph
and the sub-digraph with that set of arcs when the set of nodes is not
ambiguous.

\subsection{Omission Schemes}
In this section, we introduce our model and the associated notation. 
Communication in our model is synchronous but not reliable, and communication
is performed in rounds.    
Communication with omission faults is described by a spanning sub-graph 
of ${G}$ with the semantics that are specified in Section~\ref{execution}. 

Throughout this paper, the underlying graph $G=(V,E)$ is fixed, and we
define the set $\Sigma=\{(V,{E'}) \mid {E'}\subseteq{E}\}.$
This set represents all possible simultaneous communications given the
underlying graph $G$.

\begin{definition}
  An element of $\Sigma$ is called a \emph{communication event} (or
  \emph{event} for short). 
  An \emph{omission scenario} ( or \emph{scenario} for short) is an
  infinite sequence of communication events.  
  An \emph{omission scheme} over $G$ is a set of omission scenarios.
\end{definition}

A natural way to describe communications is to consider $\Sigma$ to be
an alphabet, with communication events as letters of the alphabet,
and scenarios as infinite words.
We will use standard
concatenation notation when describing sequences.
If $w$ and $w'$
are two sequences, then $ww'$ is the sequence that starts with the
ordered sequence of events
$w$ followed by the ordered sequence of events
$w'$. This notation is extended to sets in an obvious way. 
The empty word is denoted $\varepsilon$.
We will use the following standard notation to describe
our communication schemes.
\begin{definition}[\cite{PPinfinite}] 
  Given $R\subset\Sigma$, $R^*$ is the set of all finite
  sequences of elements of $R$, and $R^\omega$ is the set
  of all infinite ones.%
\end{definition}

The set of all possible scenarios on $G$ is then $\Sigma^\omega$.
A given word $w\in\Sigma^*$ is called a \emph{partial scenario} and 
$|w|$ is the length of this partial scenario. 
An {omission scheme} is then a subset $\scheme$ of $\Sigma^\omega$.
A \emph{mobile omission scheme} is a scheme that is equal to
$R^\omega$ for some subset $R\subseteq\Sigma$. 
In this paper, we consider only mobile omission schemes.
Note that we do not require $G$ to belong to $R$.
A formal definition of an execution subject to a scenario will be
given in Section~\ref{execution}.  
Intuitively, the $r$-th letter of a scenario will describe which
communications are reliable during round $r$.
\medskip

Finally, we recall some standard definitions for infinite words and
languages over an alphabet $\Sigma$.
Given $w=(a_1,a_2,\ldots)\in \Sigma^\omega$, a \emph{subword} of $w$ is a
(possibly infinite) sub-sequence $(a_{\sigma(1)},a_{\sigma(1)},\ldots)$,
where $\sigma$ is a strictly increasing function. 
A word $u\in\Sigma^*$ is a prefix of $w\in\Sigma^*$
(resp.\ $w'\in\Sigma^\omega$) if there exists $v\in\Sigma^*$
(resp.\ $v'\in\Sigma^\omega$) such that $w=uv$ (resp.\ $w'=uv'$).
Given $w\in\Sigma^\omega$ and $r\in\N$, $w|_r$ is the finite prefix of $w$
of length $r$.

\begin{definition}
  Let $w\in\Sigma^\omega$ and $L\subset\Sigma^\omega$. Then
  $\mathit{Pref}(w)=\{u\in\Sigma^*|u \mbox{ is a prefix}$ $\mbox{of } w\}$, and
  $\mathit{Pref}(L)=\mathop\bigcup\limits_{w\in L}\mathit{Pref}(w)$.
  A word $w'$ is an \emph{extension} of $w$ in $L$, if $ww'\in L$.
\end{definition}

\subsection{Examples}
  We do not restrict our study to regular sets, however all omission
  schemes known to us are regular, including the
  following examples, so we will use the notation for regular sets.
  We present examples for systems with two processes but they can be
  easily extended to any arbitrary graph.
  The set
  $\Sigma=\{\lok,\lblanc,\lnoir,\lall\}$
  is the set of directed graphs with two nodes \blanc and \noir.
  The subgraphs in $\Sigma$ describe what can happen during a given round
  with the following interpretation:

  \begin{compactitem}
  \item
    $\lok:$ all messages that are sent are correctly received;
  \item
    $\lblanc:$ the message from process \blanc, if any, is not received;
  \item
    $\lnoir:$ the message from process \noir, if any, is not received;
  \item
    $\lall:$ no messages are received.

  \end{compactitem}

  \begin{example}\label{2processes}
    The set $\{\lok\}^\omega$ corresponds to a reliable system. The set
    $\mathcal O_1 = \{\lok,\lblanc,\lnoir\}^\omega$ is well studied 
    and corresponds to the 
    situation in which there is at most one omission per round.
  \end{example}

\begin{example}\label{contrex}
  The set $\mathcal H=\{\lblanc,\lnoir\}^\omega$ describes a system in
  which at most one
  message can be successfully received in any round, 
  and if only one message is sent, it might not be received.
\end{example}

The examples above are examples of mobile omission schemes. The
following is a typical example of a non-mobile omission scheme.

\begin{example}\label{ex:crash}
  Consider a system in which at most one of the processes can crash.
  From the communications point of view, this is equivalent to a system
  in which it is possible that no messages are transmitted by one of
  the processes after some arbitrary round.
The associated omission scheme is the following:
  $$\mathcal C_1=\{\lok^\omega\}\cup\{\lok\}^*(\{\lblanc^\omega\}\cup 
  \{\lnoir^\omega\}).$$
\end{example}

\subsection{Reliable Execution of a Distributed Algorithm Subject to Omissions}
\label{execution}

Given an omission scheme $\scheme$, we define  what is a successful
execution of
a given algorithm \algo with a given initial configuration $\iota$. 
Every process can execute the following communication primitives:
\begin{compactitem}
\item $send(v,msg)$ to send a message $msg$ to an out-neighbour $v$,
\item $recv(v)$ to receive a message from an in-neighbour $v$.
\end{compactitem}

An \emph{execution}, or \emph{run}, of an algorithm  \algo \emph{subject to}
scenario $w\in \scheme$  is the following. Consider process $u$ and
one of its out-neighbours $v$.
During round $r\in\N$, a message $msg$ is sent from $u$ to $v$,
according to algorithm \algo.  
The corresponding $recv(v)$ will return
$msg$ only if $E'$, the $r$-th letter of $w$, is such that $(u,v)\in
E'$. Otherwise the returned value is $null$.
All messages sent in a round can only be received in the same round.
After sending and receiving messages,
all processes update their states according to \algo and the messages
they received. 
Given $u\in \mathit{Pref}(w)$, let $s^\proc(u)$ denote the
state of process \proc at the end of the $|u|$-th round of algorithm \algo
subject to scenario $w$.
The initial state of \proc is $\iota(x)=s^\proc(\varepsilon)$.
A \emph{configuration} corresponds to the collection of local states at
the end of a
given round. 
An \emph{execution} of \algo subject to $w$ is the (possibly infinite)
sequence of such  message exchanges and corresponding configurations. 

\begin{remark}
  With this definition of execution, the environment is independent of the
  actual behaviour of the algorithm, so communication failures do
  not depend upon whether or not messages are sent.
  This model is not suitable for
  modelling omissions caused by congestion. See
  \cite{krlovi_broadcasting_2003} for examples of threshold-based
  omission models.   
\end{remark}

\begin{definition}
A algorithm \algo \emph{solves a problem} \prob \emph{subject to
  omission scheme $\scheme$} with initial configuration $\iota$,
if, for any scenario $w\in \scheme$, there exists $u\in \mathit{Pref}(w)$ such
that the state $s^x(u)$ of each process $x\in V$ satisfies the
specifications of \prob for initial configuration $\iota$. 
In such a case, \algo is said to be \emph{\scheme-reliable} for \prob. 
\end{definition}

\begin{definition}
  If there exists an algorithm that solves a problem \prob subject to
  omission scheme \scheme, then we say that \prob is
  $\scheme-$solvable. 
\end{definition}

\begin{remark}
  We emphasize that for an algorithm, ``knowing'' the omission
  scheme against which it runs is not the same as knowing whether or not a
  given message is actually received.   
\end{remark}

\section{The Problems}
\label{theproblems}

\subsection{The Binary Consensus Problem}
\label{defconsensus}
A set of synchronous processes wishes to agree about a binary
value. This problem was first identified and formalized by Lamport,
Shostak and Pease \cite{LSP}. Given a set of processes, a consensus
protocol must satisfy the following 
properties for any combination of initial values~\cite{LynchDA}:

\begin{compactitem}
\item \emph{Termination}: every process decides some value;
\item \emph{Validity}: if all processes initially propose the same value $v$, 
then every process decides $v$;
\item \emph{Agreement}: if a process decides $v$, then every process decides
 $v$.
\end{compactitem}

Consensus with these termination and decision requirements
is more precisely referred to as \emph{Uniform Consensus}
(see \cite{RaynalSynchCons} for a discussion).
Given a fault environment, the natural questions are: is
Consensus solvable, and if it is solvable, what is the minimum
number of rounds to solve it?  

\subsection{The Broadcast Problem}
\label{defbroadcast}
  Let $G=(V,E)$ be a graph. There is a \emph{broadcast algorithm
    from} $u\in V$, if there exists an algorithm that can successfully
  transmit any value stored in $u$ to all nodes of $G$.
   The \emph{Broadcast problem on  graph
   $G$} is to find a $u\in V$ and an algorithm \algo such that \algo
   is a broadcast algorithm from $u$.  
   Given an omission scheme \scheme on $G$, $G$ is
   \scheme-broadcastable if there exists a $u\in V$ such that there is an
   \scheme-reliable broadcast algorithm \algo from $u$. 

\subsection{First Reduction}
\label{1streduc}
The next proposition is quite well known but leads to very interesting
questions.  
\begin{proposition}
  Let $G$ be a graph and \scheme an omission scheme for $G$. If $G$ is
  \scheme-broadcastable, then Consensus is \scheme-solvable on $G$.
\end{proposition}
\begin{proof}
  If $G$ is \scheme-broadcastable then there exists a node $u$ and 
  an algorithm \algo for broadcasting any value from $u$ subject to
  $\scheme$. 
  The consensus algorithm uses \algo to broadcast the initial value of
  $u$, and then every node decides the value received from $u$.
  As the value that $u$ broadcasts is one of the initial values,
  the algorithm satisfies all three conditions, and is obviously
  \scheme-reliable. 
\hfill$\qed$\end{proof}

We now present an example that shows that the converse is not always true.

\begin{example}\label{contrex-algo}
  The omission scheme $\mathcal H = \{\lblanc,\lnoir\}^\omega$ of
  Example~\ref{contrex} is an example of 
  a system for which there is a Consensus algorithm but no
  Broadcast algorithm.

  It is easy to see that it is not possible to broadcast from \blanc
  (resp.\ \noir) subject to $\mathcal H$ because $\lblanc^\omega$ 
  (resp.\ $\lnoir^\omega$) is a possible scenario. 
  However, the following one-round algorithm (the same for both processes)
  is an $\mathcal H-$reliable Consensus algorithm:
  \begin{compactitem}
  \item send the initial value;
  \item if a value is received, decide this value, otherwise decide the 
    initial value.
  \end{compactitem}
This algorithm is correct, as exactly one process will receive a value,  
but it is not possible to know in advance whose value will be received.
\end{example}

We propose to study the following question: when is the solvability of
Consensus equivalent to the solvability of Broadcast? That is, given a
graph $G$, what are the mobile omission schemes \scheme on $G$ such that
Consensus is \scheme-solvable and
$G$ is \scheme-broadcastable.  
In the process of answering this question, we will give a simple
characterization of the solvability of Broadcast and a necessary
condition for the solvability of Consensus
subject to mobile omission schemes.

\section{Broadcastability}

\subsection{Flooding Algorithms}
We start with a basic definition and lemma.  
\begin{definition}
  Consider a sub-digraph $H$ of $G$ and a node $u\in V$.
  A node $v\in V$ is \emph{reachable from $u$ in $H$} if 
  there is a directed path from $u$ to $v$ in $H$.
  Node $u$ is a \emph{source} for $H$ if every $v\in V$ is
  reachable from $u$ in $H$.
\end{definition}

In a \emph{flooding algorithm}, one node repeatedly sends a message
to its neighbours, and each other node repeatedly forwards any
message that it receives to its neighbours. The following
useful lemma (from folklore) about synchronous flooding algorithms
is easily extended to the
omission context. Let $\mathcal F_u^r$ denote a flooding algorithm
that is originated by $u\in V$ and that halts after $r$ synchronous
rounds.  

\begin{lemma}\label{obvious}
  A node $u\in V$ is a \emph{source} for $H$ if and only if for all $r\geq|V|$, 
  $\mathcal F_u^r$ is $H^\omega -$reliable for the Broadcast problem. 
\end{lemma}

\subsection{Characterizations of Broadcastability with Arbitrary Omissions}
\label{sec:broadcastability}

We have the following obvious but fundamental lemma. We say that a node
is \emph{informed} if it
has received the value from the originator of a broadcast.
\begin{lemma}\label{subwordflood}
  Let $u\in V, r\in\N$, and let $\mathit{Inform}(w)$ be the set of nodes
  informed by
  $\mathcal F_u^r$ under the partial execution subject to $w\in\Sigma^*.$
  Then for any subword $w'$ of $w$, $\mathit{Inform}(w')\subseteq \mathit{Inform}(w).$   
\end{lemma}

\begin{theorem}\label{broadcastability}
  Let $G$ be a graph and $R$ a set of communication events for $G$. 
  Then $G$ is $R^\omega-$broad\-castable if and only if there exists $u\in V$ 
  that is a source for all $H\in R$.
\end{theorem}
\begin{proof}
  In the first direction, suppose that we have a broadcast algorithm
  from a given $u$ that is $R^\omega-$reliable. Then an execution
  subject to $H^\omega$ is successful for any $H\in R$,
  so $u$ is a source for $H$ by Lemma~\ref{obvious}.

  In the other direction, choose the flooding
  algorithm $\mathcal F_u^{|R|\times|V|}$ to be the broadcast algorithm
  and consider a scenario $w\in R^{|R|\times|V|}.$ 
  There is an event $H\in R$ that appears at least $|V|$ times in
  $w$, hence $H^{|V|}$ is a subword of $w$. 
  As $u$ is a source for $H$ by Lemma~\ref{obvious},
  $\mathit{Inform}(H^{|V|})=V$.
  By Lemma~\ref{subwordflood}, $\mathit{Inform}(w)=V$, and the
  flooding algorithm is $R^\omega-$reliable. 
  \hfill$\qed$\end{proof}

\begin{definition}
  The \emph{set of sources of an event $H\in\Sigma$} is
  $B(H)=\{u\in V \mid u \mbox{ is a source for } H\}.$
\end{definition}

\begin{definition}
  Let $H_1,\dots,H_q\in\Sigma$. Then the set $\{H_1,\dots,H_q\}$ is
  \emph{source-incom\-pa\-tible} if 
 \begin{compactenum}
 \item 
$\forall 1\leq i\leq q, B(H_i) \neq \emptyset,$
 \item 
$\mathop{\bigcap}\limits_{1\leq i\leq q} B(H_i) = \emptyset$.
 \end{compactenum}
  \label{intersection}
\end{definition}

With these definitions we can restate the Broadcastability theorem
(Theorem~\ref{broadcastability}):

\begin{theorem}\label{equivb}
  Let $G$ be a graph and $R$ a set of communication events for $G$. 
  Then $G$ is
  $R^\omega-$broad\-castable if and only if every event
  in $R$ has a source, and $R$ is not source-incompatible.
\end{theorem}

\subsection{A Converse Reduction}

\macro{\dG}{R}
\macro{\source}{B}
Given a subset $X\subset V$ of vertices, and $H$ an event we denote 
$In_X(H) = \{(v,u)\in H \mid u\in X\}$.
Given \dG an omission scheme.
We now define a more precise relation about indistinguishability.

\begin{definition}
  Given three directed graphs $G,H,K\in\dG$,%
  we define the following relation denoted by $G\alpha_{K} H$ if
  $In_{\source(K)}(G) = In_{\source(K)}(H)$.
  The relation $\alpha^*$ is the transitive closure of 
  $\alpha_K$ relations for any $K\in\dG$.
  
  We denote $\beta$ the coarsest equivalence relation included
  in $\alpha^*$ such that for all graphs $G,H$
  \begin{compactdesc}
  \item[(Closure Property)]\label{closure}
$G\beta H \Longrightarrow \exists H_0,\dots,H_q$ and
    $K_1,\dots,K_{q}$ such that
    \begin{compactitem}
    \item $G=H_0,H=H_q$,
    \item $\forall i\geq1, H_i\beta G, K_i\beta G,$
    \item $\forall i\geq0, H_i\alpha_{K_i}H_{i+1}.$
    \end{compactitem}
  \end{compactdesc}
\end{definition}

The relation $\alpha_K$ describes how some communication events are
indistinguishable to the all the nodes of $\source(K)$. 
The relation $\beta$ is well defined as the equality relation satisfies
such a closure property. And for any two relations $R_1$ and $R_2$ that
satisfy the property, we have $R_1\cup R_2$ that satisfies the Closure.

\begin{example}
  In $\mathcal O_1$ from Example~\ref{2processes}, there is only one
  equivalence class. Let's see why. First, the sets of sources to
  consider are:
  \begin{compactitem}
  \item $B(\lblanc) = \{\noir\};$
  \item $B(\lnoir) = \{\blanc\};$
  \item $B(\lok) = \{\blanc,\noir\}.$
  \end{compactitem}
  We have $\lblanc\;\alpha_{\{\blanc\}}\;\lok$ and 
  $\lnoir\;\alpha_{\{\noir\}}\;\lok$. Therefore, all communication events
  are $\beta-$equi\-va\-lent. 
  In Example~\ref{contrex}, $\beta$ has two equivalence classes,
  and every node can distinguish immediately which communication event
  happened.  
\end{example}
  As will be seen later in Section~\ref{equivalence}, the omission
  schemes in \cite{SW07} and \cite{SWK09}, and more generally, all
  schemes that are defined by bounding the number of omissions in some way, 
  have only one $\beta-$class when they are source-incompatible.
Finally,
we can now state the main theorem.
\begin{theorem}\label{impossibility}
  Let $G$ a graph and $R$ a set of communication events for $G$. 
  If Consensus is
  $R^\omega-$sol\-vable then for every $\beta_R-$class $C$, $G$ is
  $C^\omega-$broadcastable.
\end{theorem}

\section{Proof of Main Theorem}
\label{proof}
\subsection{Events without Sources}
First we consider the cases in which there are events without sources.

\begin{proposition}\label{nonreachable}
  If there is an $H\in R$ that has no source, then Consensus
  is not $R^\omega-$solvable.
\end{proposition}
\begin{proof}
  If $H$ has no source, then there exist two non-overlapping,
  non-empty subsets of nodes $U_0$ and $U_1$ such that there are no
  paths in $H$ from $V\backslash U_i$ to $U_i$, $i=0,1$.

  We consider the three following initial configurations:
  \begin{compactenum}
  \item $I_0$: initial value is $0$ at every node,
  \item $I_1$: initial value is $1$ at every node,
  \item $I$: initial value is $0$ if and only if the node belongs to
    $U_0.$ 
  \end{compactenum}

  Under scenario $H^\omega$, $I$ is not distinguishable from $I_0$
  (resp.\ $I_1$) for $U_0$ (resp.\ $U_1$). So subject to $H^\omega$, the
  algorithm decides $0$ in $U_0$ and $1$ in $U_1$, and this contradicts the
  Agreement property. 
\hfill$\qed$\end{proof}

The proof of the main theorem uses an approach that is similar to the
adjacency and
continuity techniques of \cite{SW07}. 
So, we will first prove these two properties. What should be noted
is that the adjacency and continuity properties are mainly consequences of
the fact that the scheme is a mobile scheme.

\subsection{An Adjacency Property}

\begin{lemma}\label{bstable}
  Let $H$ be a subgraph of $G$, and let $(s,t)\in H$. If $t\in B(H)$ then
  $s\in B(H)$.
\end{lemma}

\begin{proposition}\label{adjacency}
  Let $H\in R$ and $w,w'\in R^*$ such that $|w|=|w'|$ and $s^p(w)=s^p(w')$
  for all $p\in B(H)$. Then for all $k\in\N$ and all $p\in B(H)$,
  $s^p(wH^k)=s^p(w'H^k)$. 
\end{proposition}
\begin{proof}
  The proof relies upon Lemma~\ref{bstable} which implies that
  processes from $B(H)$ can only receive information from $B(H)$ under
  scenario $H^k$, for any $k\in\N$. 
\hfill$\qed$\end{proof}

\subsection{A Continuity Property}

\begin{lemma}\label{indist}
  Let $H,H'\in R$ such that $H\alpha_{B(A)}H'$ for some $A\in
  R$. Then for all $w\in R^\omega$ and all $p\in B(A)$, $s^p(wH)=s^p(wH')$.
\end{lemma}
\begin{proof}
  By definition of $\alpha_U$ relations, processes in $B(A)$
  cannot distinguish $H$ from $H'$ meaning they are receiving
  the exact same messages from exactly the same nodes in both
  scenarios.
  Hence they end in the same states.   
\hfill$\qed$\end{proof}

\begin{proposition}\label{continuity}
  Let $H, H' \in R$ such that $H\beta H'.$ Then for every $w\in
  R^*$, there exist $H_1,\dots,H_q$ in the $\beta-$class of $H$
  and $H'$, and $A_0,\dots,A_{q}\in R$, such
  that for every $0\leq i\leq q$ and every $p\in B(A_i)$, 
  $s^p(wH_i)=s^p(wH_{i+1})$, where $H_0=H$ and $H_{q+1}=H'$.  
\end{proposition}
\begin{proof}
  By Lemma~\ref{indist} and definition of $\beta$.
\hfill$\qed$\end{proof}

\subsection{End of Proof of Theorem \ref{impossibility}.}

  We will use a standard bivalency technique. We suppose that we
  have an algorithm that solves Consensus. 
  A configuration is said to be $0-$valent (resp.~$1-$va\-lent) if all
  extensions decide $0$ (resp.~$1$). 
  A configuration is said to be bivalent subject to $L$ if there
  exists an extension in $L$ that decides 0 and another extension in $L$
  that decides 1. 

  \begin{lemma}[Restricted Initial Bivalent
    Configuration]\label{bivinit} 
    If there exists a source-incompatible set $D$, then there
    exists an initial configuration that is bivalent subject to
    $D^\omega$.
  \end{lemma}

  \begin{proof}
    Suppose that $\{H_1,\dots,H_q\}$ is a source-incompatible set in $D$.
    There exist disjoint non-empty sets of nodes
    $M_1,\dots,M_k$ such that $\forall i, \exists I\subset[1,k]$,  
    $B(H_i)=\mathop{\bigcup}\limits_{j\in I} M_j.$ 
    
    Consider $\iota_0$ (resp.\ $\iota_{k}$) in which all nodes of 
    $\mathop{\bigcup}\limits_{1\leq j\leq k} M_j$ have initial value
    $0$ (resp.\ $1$). The initial configuration $\iota_0$ is
    indistinguishable from the configuration in which all nodes have
    initial value $0$ for the nodes of $B(H_i)$ under scenario
    $H_i^\omega$, for every $i$. 
    Hence $\iota_0$ is $0-$valent. Similarly $\iota_k$
    is $1-$valent. We consider now the initial configurations
    $\iota_l$, $1\leq l\leq k-1$ in which all nodes from
    $\mathop{\bigcup}\limits_{1\leq j\leq k-l} M_j$
    have initial value $0$, and all other nodes have initial value
    $1$. 

    Suppose now that all initial configurations are univalent. Then
    there exists $1\leq l\leq k$ such that $\iota_{l-1}$ is $0-$valent and
    $\iota_l$ is $1-$valent. As the set is source-incompatible,
    there must exist $i\in[1,q]$ such that $M_l\cap B(H_i) =
    \emptyset.$ 
    So, we can apply Proposition~\ref{adjacency} to $H_i$.
    This means that all nodes in $B(H_i)$ decide the same value for
    both initial configurations, $\iota_{l-1}$ and $\iota_l$, 
    under scenario $H_i^\omega$, and this is a contradiction. 
  \hfill$\qed$\end{proof}

  \begin{lemma}[Restricted Extension]\label{extcont}
    Let $C$ be a $\beta-$class.
    Every bivalent configuration in $C^\omega$ has a succeeding
    bivalent configuration in $C^\omega$.
  \end{lemma}
  \begin{proof}
    Consider a bivalent configuration obtained after a partial
    execution subject to $w\in C^*$.
    By way of contradiction, suppose that all succeeding configurations
    in $C^\omega$ are
    univalent. Then there exist succeeding configurations
    $wH$ and $wH'$ that are respectively
    $0-$valent and $1-$valent, as $w$ is bivalent. 
    
    By Proposition~\ref{continuity}, there exist $H_1,\dots,H_q$ in $C$
    and $A_0,\dots,A_{q}\in R$  
    such that $s^p(wH_i)=s^p(wH_{i+1})$ for every $0 \leq i \leq q$ and every
    $p\in B(A_i)$, where $H_0=H$ and $H_{q+1}=H'.$
    By hypothesis, all succeeding configurations
    $wH_i$ are univalent. As
    $H\alpha_{B(A_0)}H_1$, we get that
    processes in $B(A_0)$ are in the same state after $H$ and after
    $H_1$. Hence, by Proposition~\ref{adjacency}, they are also in the
    same state after $HA_0^k$ and after $H_1A_0^k$, so they
    decide the same value and $wH_1$ is $0-$valent. We can repeat
    this for any $1\leq i\leq q$. Hence $wH'$ is also $0-$valent, a
    contradiction. 
  \hfill$\qed$\end{proof}
      
  We can now finish the proof with the standard bivalency arguments.
    Suppose that we have a source-incompatible set in the same
    $\beta-$class $C$.
    Also suppose that there exists an $R^\omega-$reliable Consensus
    algorithm for $G$.  
    By Lemma~\ref{bivinit}, there exists an 
    initial configuration that is bivalent in $C^\omega$. 
    From Lemma~\ref{extcont}, we deduce that
    the algorithm does not satisfy the Termination property for Consensus 
    on some execution subject to $C^\omega\subset
    R^\omega$, which is a contradiction.
    Using Proposition~\ref{nonreachable} and Theorem~\ref{equivb}, we
    conclude the proof of Theorem~\ref{impossibility}.

\section{Solvability of Consensus vs Broadcast}
\label{equivalence}
In this section, we prove that the Consensus and Broadcast problems
are equivalent for the large family of omission schemes that are
defined over convex sets of events. This has very important
consequences as checking Broadcastability is quite simple
(see Theorem~\ref{broadcastability}). 

\begin{definition}\label{convexity}
  A set of communication events $R$ is \emph{convex} if, for every
  $H,H'\in R$ and every $a\in H'$, $H\cup\{a\}\in R$.  
\end{definition}

Basically, this definition says that a convex set of communication
events $R$ is closed under the operation of adding a reliable
communication event $a$ from one event $H'$ to another event $H$.
This is an important subfamily because sets of events that are defined by
bounding the number of omissions, \emph{for any way of counting
  them}, are convex.
Stated differently, adding links to an event $H$ with a bounded number
of omissions cannot result in an event with more omissions.
The convexity property does not depend upon the way that
omissions are counted.

\begin{theorem}
  Let $R\subset \Sigma$ be a convex set of communication events
  over a graph $G$. Then
  Consensus is $R^\omega-$solvable if and only if $G$ is
  $R^\omega-$Broadcastable. 
  \label{convex-consensus}
\end{theorem}
\begin{proof}
By Theorem~\ref{impossibility}, we only have to show that there is no
source-incompatible set in $R$. 
We will show that if there is such a set $\{H_1,\dots,H_q\}$, 
then there is only one $\beta_R$ class.  

There exist disjoint,
non-empty sets of nodes $M_1,\dots,M_k$ such that
$\forall i, \exists I\subset[1,k]$,  
$B(H_i)=\mathop{\bigcup}\limits_{j\in I} M_j$.
The $M_j$ are ``generators'' for the sets of sources.
We use $M_J$ to denote $\mathop{\bigcup}\limits_{j\in J} M_j$
for any $J \subset [1,k].$ 
Note that, as the intersection of the $H_i$ is empty
(Definition~\ref{intersection}),  for each $j\in [1,k]$, 
there exists $i_j$ such that $M_j\cap B(H_{i_j}) = \emptyset$.

Now, consider $H_1 \neq H_2 \in R$. 
We will show that $H_1\;\beta\; (H_1\cup H_2)$.
Using the decomposition into $M_j$s,  
there exist three mutually disjoint (possibly empty) subsets
$J_1,J,J_2$ of $[1,k]$,
such that $B(H_1) = M_{J_1\cup J}$ and $B(H_2) = M_{J_2\cup J}.$

Let $H'_1=H_1\cup\{(s,t)\in H_2 \mid t\in M_{J_1}\}$. As the
intersection of $B(H_2)$ with $M_{J_1}$ is empty, we have
$H_1 \alpha_{B(H_2)} H'_1.$

Similarly, letting $H''_1 = H'_1\cup\{(s,t)\in H_1 \mid t\in M_{J_2}\},$
we have $H'_1 \alpha_{B(H_1)} H''_1.$

To obtain $H_1\cup H_2$, we need to add to $H''_1$ the arcs with
heads in $B(H_1)\cap B(H_2)=M_J$. Let $J=\{j_1,\dots,j_q\}$ and
$K_k=H''_1\mathop{\bigcup}\limits_{l\leq k} M_{j_l}.$ 

For each $l\in J$, 
there exists $i_l$ such that $M_l\cap B(H_{i_l}) = \emptyset$.
Therefore, for all $k$, $K_{k-1} \alpha_{B(H_{i_k})} K_k$.
So finally, we obtain $H_1\;\beta\; (H_1\cup H_2)$. Similarly, 
$H_2\;\beta\; (H_1\cup H_2)$, and $H_1\beta H_2.$
\hfill$\qed$\end{proof}

Let $O_f(G)$ denote the set of communication events with
at most $f$ omissions from the underlying graph $G$. 
An upper bound on $f$ for solvability of Consensus subject to
$O_f(G)^\omega$ was given in
\cite{SW07}, and it was proved to be tight with an ad hoc technique in
\cite{2generals}. We can now state this result as an immediate corollary of
Theorem~\ref{convex-consensus}.  

\begin{corollary}
  Let $f\in\N$. Consensus is solvable subject to
  $\mathcal O_f(G)^\omega$ if and only if $f<c(G)$, where $c(G)$ is the
  connectivity of the graph $G$. 
\end{corollary}

The equivalence of Consensus and Broadcast includes the number of rounds
to solve them.
\begin{proposition}
  Let $R\subset \Sigma$ be a convex set of events. If Consensus is
  $R^\omega-$solva\-ble then it is solvable with exactly the
  same number of rounds as Broadcast subject to $R^\omega$.
\end{proposition}
\begin{proof}
  We only have to show that Consensus cannot be solved in fewer
  rounds than Broadcast. 
  Due to space limitations, we only present a sketch of the proof.
  We use the same bivalency technique as in
  Section~\ref{proof}. So, suppose that Consensus is solvable in $r_c$
  rounds while Broadcast needs more than $r_c$ rounds for any
  originator.  

  First, we show that there must be a bivalent initial configuration. 
  Let $V=\{v_1,\dots,v_n\}$ and let $\iota_l$ be the initial
  configuration in which $v_i$ has initial value $0$ if $i\leq l$.
  If all initial configurations $\iota_l$, $1\leq l\leq n$ are univalent,
  then there exists $l$ such that $\iota_{l-1}$
  is $0-$valent and $\iota_l$ is $1-$valent. As a Broadcast from $v_l$
  needs strictly more than $r_c$ rounds, there exists a vertex $v$ that
  does not receive the value from $v_l$, so 
  no executions of length $r_c$ of the
  Consensus algorithm from initial configurations $\iota_{l-1}$ and
  $\iota_l$ can be distinguished by $v$. Consequently $v$ will
  decide the same value for both initial configurations, 
  a contradiction.

  Now we show that if all extensions of a bivalent configuration
  are univalent, then the Consensus algorithm needs more
  than $r_c$ rounds to conclude. Indeed, if we have an extension,
  starting with communication event $H_0$, that is $0-$valent, and
  another extension, starting with communication event $H_1$, that is
  $1-$valent, we can repeat the above technique by adding arcs to
  $H_0$ to obtain $H_0\cup H_1$. 
  The addition of arcs can be done by grouping them according to their heads.
  If only one node has a different state for two events, 
  then it would need more than $r_c$ rounds to inform all other
  nodes. 
\hfill$\qed$\end{proof}

There are many results concerning the Broadcast problem in
special families of networks. General graphs are studied in
\cite{chlebus_broadcasting_1994} and hypergraphs are studied
in \cite{de_marco_broadcasting_1998}.  An optimal algorithm for
the family of hypercubes in given in \cite{dobrev_optimal_1999}.
For a hypercube of dimension $n$, if
at most $n-1$ messages are lost during each round, then Broadcast can be
solved in $n+2$ rounds, compared to $n$ rounds when there are no
omission faults. In \cite{dobrev_dynamic_2004}, the precise
impact on Broadcast of the actual number of faults is given. 
Based on the results in \cite{dobrev_dynamic_2004}, we get
the following bounds for Consensus.
\begin{corollary}
  In hypercubes of dimension $n$, if the global number of omissions is
  at most $f$ per round, then
  \begin{compactenum}
  \item if $f\geq n$, then Consensus is not solvable,
  \item if $f=n-1$, Consensus is solvable in exactly $n+2$ rounds,
  \item if $f=n-2$, Consensus is solvable in exactly $n+1$ rounds,
  \item if $f<n-2$, Consensus is solvable in exactly $n$ rounds.
  \end{compactenum}
\end{corollary}

The following example shows that there are mobile schemes that are
broadcastable but for
which Consensus is solvable is fewer rounds than Broadcast.
\begin{example}\label{contrex-comp}
  Let $R=\{H_1,H_2\}$ where $H_1$ (resp.\ $H_2$) is given by Fig.~1
  (resp.\ Fig.~2). One can see that $d$ needs two rounds to broadcast
  in $H_1$, and $c$ needs two rounds to broadcast in
  $H_2$. Nodes $a$ and $b$ need more than two rounds. However there is
  a Consensus algorithm that finishes in one round. Notice that every
  node can detect which of the communication events actually happened,
  so the Consensus algorithm in which every node decides
  the value from $c$ if $H_1$ happened and the value from $d$ if $H_2$
  happened uses only one round. 
\end{example}

\newcommand{\vertex}[3][]{\node[draw,circle] (#2) at (#3) {#1}}
\newcommand{\arc}[2]{\draw[->,thick] (#1) -- (#2)}
\begin{figure}[t]
  \begin{minipage}[c]{0.49\linewidth}
    \centering
    \begin{tikzpicture}
      \vertex[a]{A}{0,0}; %
      \vertex[b]{B}{0:2cm}; %
      \vertex[c]{C}{40:1.35cm}; %
      \vertex[d]{D}{-40:1.35cm}; %

      \arc{C}{A};
      \arc{C}{B};
      \arc{D}{A};
      \arc{D}{B};
      
      \arc{A}{B};
      \arc{C}{D};
      \draw[->,thick] (B) to[bend right=45] (C);
    \end{tikzpicture}
    \caption{ Event $H_1$.}
  \end{minipage} \hfill
  \begin{minipage}[c]{0.49\linewidth}
    \centering
    \begin{tikzpicture}
      \vertex[a]{A}{0,0}; %
      \vertex[b]{B}{0:2cm}; %
      \vertex[c]{C}{40:1.35cm}; %
      \vertex[d]{D}{-40:1.35cm}; %

      \arc{C}{A};
      \arc{C}{B};
      \arc{D}{A};
      \arc{D}{B};
      
      \arc{B}{A};
      \arc{D}{C};
      \draw[->,thick] (A) to[bend right=45] (D);
    \end{tikzpicture}
    \caption{ Event $H_2$.}
  \end{minipage}
  \label{quickerconsensus}
\end{figure}  

\section{Conclusions}
We have presented a new necessary condition for solving Consensus on
communication networks subject to arbitrary mobile omission faults. 
We conjecture that this condition is actually sufficient, therefore leading
to a complete characterization of the solvability of Consensus in
environments with arbitrary mobile omissions. 
For a large class of environments that includes any environment
defined by bounding the number of omissions during any round, for
\emph{any} way of counting omissions, we proved that the Consensus
problem is actually equivalent to the Broadcast problem. We also gave
examples (Ex.~\ref{contrex} and Ex.~\ref{contrex-comp}) showing how
Consensus can differ from Broadcast for some environments.

Finally, by factoring out the broadcastability properties required to
solve Consensus, we think that it is possible to extend this work to
other kinds
of failures, such as \emph{byzantine} communication faults.


\begin{thebibliography}{SWK09}

\bibitem[CBS09]{HO}
Bernadette Charron-Bost and André Schiper.
\newblock The heard-of model: computing in distributed systems with benign
  faults.
\newblock {\em Distributed Computing}, 22(1):49--71, 2009.

\bibitem[CDP94]{chlebus_broadcasting_1994}
B.~S Chlebus, K.~Diks, and A.~Pelc.
\newblock Broadcasting in synchronous networks with dynamic faults.
\newblock {\em {NETWORKS}}, 27:309---318, 1994.

\bibitem[DV99]{dobrev_optimal_1999}
Stefan Dobrev and Imrich Vrto.
\newblock Optimal broadcasting in hypercubes with dynamic faults.
\newblock {\em {I}nformation {P}rocessing {L}etters}, 71:81---85, 1999.

\bibitem[DV04]{dobrev_dynamic_2004}
Stefan Dobrev and Imrich Vrto.
\newblock Dynamic faults have small effect on broadcasting in hypercubes.
\newblock {\em Discrete Applied Mathematics}, 137(2):155--158, March 2004.

\bibitem[FG11]{2generals}
Tristan Fevat and Emmanuel Godard.
\newblock About minimal obstructions for the coordinated attack problem.
\newblock In {\em Proc. of 25th IEEE International Parallel \& Distributed
  Processing Symposium, IPDPS'2011}, 2011.

\bibitem[KKR03]{krlovi_broadcasting_2003}
R.~Královič, R.~Královič, and P.~Ruzička.
\newblock Broadcasting with many faulty links.
\newblock {\em {P}roceedings in {I}nformatics}, 17:211---222, 2003.

\bibitem[Lyn96]{LynchDA}
Nancy~A. Lynch.
\newblock {\em Distributed Algorithms}.
\newblock Morgan Kaufmann Publishers Inc., San Francisco, CA, USA, 1996.

\bibitem[MR02]{layering}
Yoram Moses and Sergio Rajsbaum.
\newblock A layered analysis of consensus.
\newblock {\em SIAM Journal on Computing}, 31(4):989--1021, 2002.

\bibitem[MV98]{de_marco_broadcasting_1998}
Gianluca~De Marco and Ugo Vaccaro.
\newblock Broadcasting in hypercubes and star graphs with dynamic faults.
\newblock {\em Information Processing Letters}, 66(6):321--326, 1998.

\bibitem[PP04]{PPinfinite}
J.E. Pin and D.~Perrin.
\newblock {\em Infinite Words}, volume 141 of {\em Pure and Applied
  Mathematics}.
\newblock Elsevier, 2004.

\bibitem[PSL80]{LSP}
L.~Pease, R.~Shostak, and L.~Lamport.
\newblock Reaching agreement in the presence of faults.
\newblock {\em Journal of the ACM}, 27(2):228--234, 1980.

\bibitem[Ray02]{RaynalSynchCons}
Michel Raynal.
\newblock Consensus in synchronous systems:a concise guided tour.
\newblock {\em Pacific Rim International Symposium on Dependable Computing,
  IEEE}, 0:221, 2002.

\bibitem[SW89]{timeisnotahealer}
Nicola Santoro and Peter Widmayer.
\newblock Time is not a healer.
\newblock In {\em STACS}, pages 304--313, 1989.

\bibitem[SW07]{SW07}
Nicola Santoro and Peter Widmayer.
\newblock Agreement in synchronous networks with ubiquitous faults.
\newblock {\em Theor. Comput. Sci.}, 384(2-3):232--249, 2007.

\bibitem[SWK09]{SWK09}
U.~Schmid, B.~Weiss, and I.~Keidar.
\newblock Impossibility results and lower bounds for consensus under link
  failures.
\newblock {\em SIAM Journal on Computing}, 38(5):1912–1951, 2009.
\newblock http://ti.tuwien.ac.at/ecs/people/schmid.

\end{thebibliography}
\end{document}